\documentclass[review]{elsarticle}
\usepackage{lineno,hyperref}
\modulolinenumbers[5]
\usepackage{color,graphicx,amsmath,amssymb,verbatim,mathrsfs,latexsym}
\usepackage{amsmath,amsfonts,amssymb,amscd,amsthm,amsbsy,bbm,epsf,calc,graphicx}
\usepackage{color}
\usepackage{datetime}
\usepackage{comment}
\usepackage{gensymb}
\usepackage{wasysym}
\newtheorem{theo}{Theorem}

\newtheorem{lemma}{Lemma}

\newtheorem{rem}{Remark}
\newcommand{\R}{\mathbb{R}}

\newcommand{\N}{\mathbb{N}}

\newcommand{\eps}{\varepsilon}
\newcommand{\pa}{\partial}
\newcommand{\Div}{\textrm{div}\,}

\newcommand{\T}{\mathbb{T}^3}
\newcommand{\Id}{\int}
\newcommand{\IId}{\iint}
\newcommand{\Aij}{A^{(ij)}}
\newcommand{\tr}{\mbox{tr}}
\let\na\nabla
\newcommand{\cH}{\mathcal{H}}
\newcommand{\cA}{\mathcal{A}}
\newcommand{\cB}{\mathcal{B}}
\newcommand{\bk}[1]{\left\langle #1 \right\rangle}

\def\fin{f_{\textnormal{in}}} 
\def\Fin{F_{\textnormal{in}}} 
\def\Tr{\textnormal{Tr}}

\begin{document}

\begin{frontmatter}

\title{Spectral gap  and exponential convergence to equilibrium for a multi-species Landau system}

\author{Maria Pia Gualdani\fnref{footnoteG}}
\address{Department of Mathematics, George Washington University, 801 22nd Street, NW Washington DC, 20052, USA}
\fntext[footnoteG]{gualdani@gwu.edu}

\author{Nicola Zamponi \fnref{footnoteZ}}
\address{Institute for Analysis and Scientific Computing, Vienna University of  
	Technology, Wiedner Hauptstra\ss e 8--10, 1040 Wien, Austria}
\fntext[footnoteZ]{nicola.zamponi@tuwien.ac.at}








\begin{abstract}
In this paper we prove new constructive coercivity estimates and convergence to equilibrium for a spatially non-homogeneous system of Landau equations with moderately soft potentials.  
We show that the nonlinear collision operator conserves each species' mass, total momentum,  total energy and that the Boltzmann entropy is nonincreasing along solutions of the system. The entropy decay vanishes if and only if the Boltzmann distributions of the single species are Maxwellians with the same momentum and energy.
A linearization of the collision operator is computed, which has the same conservation properties as its nonlinear counterpart.  We show that the linearized system dissipates a quadratic entropy, and prove existence of spectral gap and exponential decay of the solution towards the global equilibrium. As a consequence, convergence of smooth solutions of the nonlinear problem toward the unique global equilibrium is shown, provided the initial data are sufficiently close to the equilibrium. Our proof is based on new spectral gap estimates and uses a strategy similar to 
\cite{DauJueMouZam} based on an hypocoercivity method developed by Mouhot and Neumann in \cite{MouNeu}.

\end{abstract}




\end{frontmatter}

\linenumbers

\section{Introduction}

This manuscript is concerned with the Cauchy problem for a system of spatially non-homogeneous Landau equations  describing collisions in an ideal plasma mixture. The mixture is constituted by $N \ge 2$ species and each species $i=1,\ldots ,N$ has mass $m_i$ and is described by a density function $F_i(x,p,t)$ defined in the phase-space of position and momentum. 
The vector $F:= (F_1, \ldots, F_N)$ is said to be a solution to the multi-species Landau system if each $F_i$ satisfies

\begin{equation}
\label{eq}
 \left\{ \begin{array}{c} 
                         \pa_t F_i + \frac{p}{m_i} \cdot \nabla_x F_i =   \sum_{j=1}^N Q_{ij}(F_i,F_j), \\ 
                         {}\\
                         F(x,p,0) = \Fin(x,p),\\
           \end{array}  \right.
\end{equation}
with $(x,p,t)\in \mathbb{T}^3 \times \mathbb{R}^3  \times \mathbb{R}_+$. The operator  $Q_{ij}$ is the quadratic Landau collision operator defined as 
\begin{align}\label{coll_Oper}
                       Q_{ij}(F_i,F_j): =   \Div_p\int_{\R^3}A^{(ij)}\left[\frac{p}{m_i}-\frac{p'}{m_j}\right](F_j'\na F_i - F_i\na F_j')dp'.
\end{align}
Here we adopt the shortened notation $F\equiv F(x,p,t)$, $F'\equiv F(x,p',t)$. The term $A^{(ij)}\left[z\right] =\{ a^{(ij)}_{ks}(z)\}$ denotes a positive and symmetric matrix  
with real-valued entries defined as:

 \begin{equation*}
    A^{(ij)}\left[z\right]  :=C^{(i,j)} \left(\textrm{Id} - \frac{z \otimes z}{|z|^2}\right) \varphi(|z|),\;\;\; z\neq 0, \;\;\; C^{(i,j)} >0,
  \end{equation*}
  which acts as the projection operator onto the space orthogonal to the vector $z$. The function $\varphi(|z|)$ is a scalar valued function determined from the original Boltzmann kernel describing how particles interact. If the interaction strength between particles at a distance $r$ is proportional to ${r^{1-s}}$, then 
  \begin{align}\label{potential}
    \varphi(|z|) := |z|^{\gamma+2},\quad \gamma = \frac{(s-5)}{(s-1)}.
  \end{align}
  The constant $C^{(i,j)}>0$ is positive  and symmetric in $i,j$, and is proportional to the reduced mass of the system $ m_i m_j/(m_i + m_j)$.  We refer to \cite[Chapter 4]{LifPit} for a more accurate derivation and discussion of (\ref{eq}). The original Landau system with Coulomb interactions correspond to $\gamma =-3$. \\
  
 The purpose of this paper is to study the spectral gap properties of the linearized operator and to show exponential convergence towards the equilibrium as time grows. We assume throughout this manuscript that $\gamma \in [-2,1]$.  \\

  Let us summarize briefly the state of the art concerning the Cauchy problem (\ref{eq}) for the mono-species case.
  
   In the homogeneous setting, the cases of Maxwell molecules $\gamma =0$ and hard potentials $\gamma\in (0,1]$ have been well understood: existence and uniqueness of smooth regular solution and convergence towards the unique equilibrium state have been analyzed in several papers, see \cite{C14, CLX08, CLX10, DV00_I, DV00_II, MPX13, V98_Max} . For the spatially non-homogeneous case we refer to Alexandre and Villani \cite{AV04} for existence of renormalized solutions, to Desvillette-Villani \cite{DevVil05} for {\em conditional almost exponential} convergence towards equilibrium and to a recent work by Carrapatoso, Tristani and Wu \cite{CTW} for exponential decay towards equilibrium when initial data are close enough to equilibrium.

  The case  of soft potentials  has been proven to be harder.  For moderately soft-potentials $\gamma \in [-2,0)$ existence and uniqueness of spatially homogeneous solutions have been proven by Fournier and Guerin \cite{FG09} and by Guerin \cite{G03} using a probabilistic approach, as well as by Wu \cite{Wu14} and by Alexandre, Liao and Lin \cite{ALL15}.  Carrapatoso, Tristani and Wu \cite{CTW} recently showed exponential decay estimates for the linearized semigroup and  constructed solutions in a close-to-equilibrium regime to the non-linear inhomogeneous equation. The proof in \cite{CTW} is based on an abstract method developed by the first author and collaborators in \cite{GMM13}. 
  
  Global well-posedness theory is still missing for the Coulomb case $\gamma =-3$. For the homogeneous setting, Arsenev-Peskov \cite{AP77} showed existence of weak solutions, uniqueness was later proved by Fournier \cite{F10}. Villani \cite{V98} proved existence of a new class of solutions, the so called H-solutions, which are defined via the $L^1$- bound in time of the entropy production. Recently Alexander, Liao and Lin \cite{ALL15} gave a proof of existence of weak solutions in weighted $L^2$-space under smallness assumption on initial data.  Desvillettes \cite{D15} showed that the $H$-solutions are indeed weak-solutions since they belong to some weighted $L^1_tL^p(\R^3)$-space and Carrapatoso, Desvillettes and He \cite{CDH} have proved time convergence to the associate equilibrium at some explicitly computable rate.  For the inhomogeneous setting, Guo \cite{Guo} and Strain, Guo \cite{SG06, SG08} developed an existence and convergence towards equilibrium  theory based on energy methods for initial data close in some Sobolev norm to the equilibrium state. Recently the set of initial data for which this theory is valid has been improved by Carrapatoso and Mischler \cite{CM} via a linearization method.

Recently the first author and Guillen have shown, for the Coulomb case, global in time existence of classical solution for a modified {\em isotropic} homogeneous Landau equation 
$$
\partial_t F = \textrm{div} ( a[F] \nabla F - F \nabla a[F]),
$$
in the case of radially symmetric (but no smallness assumptions!) initial data  \cite{GG14}. Moreover, using the theory of $A_p$ weights, they showed that solutions to the original Landau equations with general initial data for $\gamma >-2$ have an instantaneous regularization which does not deteriorate as time increases, with bounds that only depend on the physical quantities, mass, momentum and energy \cite{GG16}.  \\

We believe that this is the first work that concerns system (\ref{eq}) and its linearized version. The aim of this work is to extend the spectral analysis valid for the mono-species operator to the multi-species operator with different particles' mass.  From a different prospective, the second author and collaborators have recently studied a system of Boltzmann equations for mixtures of mono-atomic particles with same mass in the case of hard and Maxwellian potentials \cite{DauJueMouZam}: the authors show an explicit spectral-gap estimate for the linearized collision operator and prove the exponential decay of the solutions  towards the global equilibrium by generalizing the hypocoercivity method developed by Mouhot and Neumann in \cite{MouNeu} for the mono-species case to the multi-species case.

\subsection{Main results}

 The main goal of this paper is to give a constructive proof of exponential
decay rate for solutions to the linear system 
\begin{equation}\label{eq.lin}
	     \left \{ \begin{array}{rl}
	      \pa_t f_i + \frac{p}{m_i}\cdot\nabla_x f_i & = \sum_{j=1}^N L_{i,j}(f), \quad  i=1,\ldots,N,\\
 f(x,p,0) &= \fin(x,p),
	     \end{array}\right.	
	\end{equation}
with 
\begin{align}
 \label{L.1_I}
 L_{i,j}(f_i,f_j) :=&\frac{1}{\sqrt{M_{i}}} \left(  Q_{ij}( \sqrt{M_i}f_{i}, M_j) +   Q_{ij}(M_i,  \sqrt{M_j}f_{j})\right) \nonumber\\
 =&\frac{1}{\sqrt{M_{i}}}\Div_p\Id\sqrt{M_i M_j'}A^{(ij)}\left[\frac{p}{m_i}-\frac{p'}{m_j}\right]\cdot \\
 &\qquad \cdot \Big(\sqrt{M_j'}\na f_i - \sqrt{M_i}\na f_j' - f_i\na\sqrt{M_j'} + f_j'\na\sqrt{M_i} \Big) dp' ,\nonumber
\end{align}
obtained from (\ref{eq}) via the perturbative expansion $ F_i = M_i + \sqrt{M_i}f_{i}$, with $M_i$ the Maxwellian equilibrium of the $i^{th}$ species 
$$ M_{i}(p) := \frac{\rho_{i}}{(2\pi m_{i}k_{B}T)^{3/2}}e^{-\frac{1}{2}\frac{|p|^{2}}{ m_{i}k_{B}T }} ,$$
where $k_B$ denotes the Boltzmann's constant and $T$ the temperature. The explicit computations of the linearization $L_{i,j}$ are outlined before Theorem \ref{theorem:linear_conservation}.  


We will  show that any solution to (\ref{L.1_I}) converges
exponentially fast to the global equilibrium.  The rate of decay is computed explicitly, following an approach already used by the second author and collaborators in 
\cite{DauJueMouZam}, which is based upon an abstract method by Mouhot and Neumann \cite{MouNeu}. 

The starting point is the existence of spectral gap for the mono-species linearized collision operator. By exploiting the symmetry properties
of the operator we are able to bound the cross terms by relating  them with the differences of momentum and energy. Hence a spectral gap
for the multi-species linearized operator follows. The hypocoercivity method by Mouhot and Neumann \cite{MouNeu} yields convergence to global equilibrium for the
solution to the in-homogeneous linearized system.  \\




Define with $L := (L_1, L_2, ..., L_N)$ the vector with components $L_i =\sum_{j=1}^N L_{i,j}$ with $L_{i,j}$ as in (\ref{L.1_I}), and by $T:= (T_1, T_2, ..., T_N)$ the transport operator, $T_i f = \frac{p}{m_i}\cdot\nabla_x f_i$.  We also denote by $\Gamma_{i}(f_i,f_j)$ the quadratic nonlinear term
\begin{align}
 \label{Gamma.1_I}
\Gamma_i(f,f) =  \frac{1}{\sqrt{M_{i}}}  \sum_{j=1}^N Q_{ij}( \sqrt{M_i}f_{i},  \sqrt{M_j}f_{j}).
 \end{align}

Let $\cH$ be the space of all functions $f= (f_1, f_2, ...,f_N) $ such that the following norm is finite: 
\begin{align*}
 \|f\|_\cH^2 := \sum_{i=1}^N 
\; & \|\bk{p}^{\gamma/2}P\na f_i\|_{L^2(\R^3, dp)}^2 + \|\bk{p}^{(\gamma+2)/2}(\mathbb{I}-P)\na f_i\|_{L^2(\R^3, dp)}^2 +\\
 &+ \|\bk{p}^{(\gamma+2)/2} f_i\|_{L^2(\R^3, dp)}^2,
\end{align*}
where $\bk{p}:= \sqrt{1+|p|^2}$ and $P := \frac{p\otimes p}{|p|^2}$. We denote by $L^2(\R^3, dp)$ all square integrable functions in the $p$-variable and with an abuse of notation we say that $f = (f_1, f_2, ...,f_N) \in L^2(\R^3, dp)$ if 
$$
\|f\|^2_{L^2(\R^3, dp)}= \sum_{i=1}^N \|f_i\|_{L^2(\R^3, dp)}^2<+\infty.
$$

Note that $\cH$ is a Hilbert space which embeds continuously into $L^2(\R^3, dp)$.\\

Our main results are summarized below.
\begin{theo}\label{thr.spec.gap}
 There exists an explicitly computable constant $\lambda>0$ such that:
\begin{align*}
 -(f,L f)_{L^2(\R^3, dp)} \geq \lambda \|f - \Pi^L f\|_{\cH}^2,\qquad f\in D(L),
\end{align*}
where $\Pi^L$ is the projection operator on the kernel $N(L)$ of $L$.
\end{theo}

The starting point of the proof of Theorem~\ref{thr.spec.gap} is a coercivity estimate for the part of the operator $L$ that describes collisions among particles of the same species. Let us denote with $L^{m} \equiv (L_{11},\ldots,L_{NN})$ and with $\Pi^{m}$ the projection operator onto the null space of $L^m$, $N(L^{m})$. Estimates of the form
$$
 C_\gamma\|f - \Pi^{m} f\|_{\cH}^2 \ge  -(f,L^{m} f)_{L^2(\R^3, dp)} \geq \lambda_m\|f - \Pi^{m} f\|_{\cH}^2,\qquad f\in D(L^{m}),
$$
have been proven in \cite{BM05, DL97, Guo, M06, MS07}.  Hence the resolvent of $L^m$ is compact for $\gamma +2 \ge 0$ and there exists a spectral gap in $L^2$ for $\gamma \ge -2$.



The second step in the proof consists in bounding the contribution of $f^{\perp}\equiv f - \Pi^{m}f$ inside the quadratic form $-(f, L^{b}f)_{L^2(\R^3, dp)}$, 
where $L^{b}\equiv L - L^{m}$ describes collisions between particles of different species:
\begin{align*}
& -(f^{\perp}, L^{b}f^{\perp})_{L^2(\R^3, dp)} \leq C_{1} \|f^{\perp}\|_{\cH}^2.
\end{align*}

In the third step, the contribution of $f^{\parallel}\equiv \Pi^{m}f$ inside the quadratic form $-(f, L^{b}f)_{L^{2}(\R^3, dp)}$ is bounded from below by the
differences of momentum $u_{i}-u_{j}$ and differences of energies $e_{i}-e_{j}$:
\begin{align*}
 -(f^\parallel,L^b f^\parallel)_{L^2(\R^3, dp)} &\geq C_2 \sum_{i,j=1}^N\left( |u_i-u_j|^2 + (e_i-e_j)^2 \right),
 \qquad f\in D(L). 
\end{align*}
This result is obtained by exploiting the structure of $N(L^{m})$.

Finally, for the fourth and last step we recall an estimate from \cite{DauJueMouZam}, which relates  $u_{i}-u_{j}$ and $e_{i}-e_{j}$ to the $\cH$ norms of $f - \Pi^L f$ and $f - \Pi^m f$ for each $f\in D(L)$:
 \begin{align}\label{lemma_no_proof}
  \sum_{i,j=1}^N\left( |u_i-u_j|^2 + (e_i-e_j)^2 \right)\geq C_3\left( 
  \|f - \Pi^L f\|_{\cH}^2 - 2\|f - \Pi^m f\|_{\cH}^2 \right).
 \end{align}
Estimate (\ref{lemma_no_proof}) was previously obtained in \cite{DauJueMouZam} for $f$ solution to a Boltzmann system. The proof is based  on a careful analysis of the different structures of $N(L^{m})$ and $N(L)$ for the Boltzmann equation, which is intimately connected to its conservation laws. Since the kernel of the Landau operator has the same structure as its Boltzmann counterpart, we refer to \cite[Lemma 15]{DauJueMouZam} for the proof of (\ref{lemma_no_proof}).


Finally, the non-positivity of $L^{b}$ allows us to write 
\begin{align*}
-(f,Lf)_{L^{2}(\R^3, dp)} &= -(f,L^{m}f)_{L^{2}(\R^3, dp)} -(f,L^{b}f)_{L^{2}(\R^3, dp)}\\ &\geq -(f,L^{m}f)_{L^{2}(\R^3, dp)} -\eta(f,L^{b}f)_{L^{2}(\R^3, dp)}
\end{align*}
for an arbitrary $\eta\in(0,1]$. Putting together
the results obtained in the previous four steps and choosing $\eta$ small enough yield the desired spectral gap, concluding the proof of Theorem \ref{thr.spec.gap}.

\begin{theo}\label{thr.conv}
  Let $f^\infty$ be the global equilibrium of the system (\ref{eq.lin}), that is, $f^\infty = \Pi^{L-T} f = \Pi^{L-T} \fin$ 
where $\Pi^{L-T}$ is the projection operator on the kernel $N(L-T)$ of $L-T$.
There exist explicitly computable constants $\tau>0$, $C>0$ such that:
 \begin{align}
  \|f - f^\infty\|_{H^1(\T \times\R^3)}\leq C \;e^{-t/\tau},\qquad t>0. \label{conv}\\\nonumber
 \end{align}

 Let  $\mathcal{M}(p) $ be the equilibrium state to (\ref{eq}) uniquely determined by the mass, first and second momentum of the initial data. 
 Assume there exists an $\varepsilon >0$ such that 
 \begin{align*}
 \left \|\frac{1}{\sqrt{\mathcal{M}}}(\Fin - \mathcal{M} ) \right\|_{H^{k}(\T \times\R^3 )} \le  \varepsilon,
 \end{align*}
 {with $k\geq 4$ }
then the nonlinear problem (\ref{eq}) has an unique solution $F(x,p,t)$ which decays exponentially fast towards the global equilibrium with a constant rate that only depends on the linearized part of the operator :
 \begin{align*}
 \left\|\frac{1}{\sqrt{\mathcal{M}}}(F - \mathcal{M} ) \right\|_{H^{k}(\T \times\R^3 )} \le  C_{\textrm{in}}\; \varepsilon \; e^{-\lambda t/4}, \qquad t>0.
  \end{align*}
The explicit value of $\lambda$ is computed in Theorem \ref{thr.spec.gap}.
\end{theo}

\begin{rem} {{The global equilibrium states  $\mathcal{M}(p) $ and $f^\infty(p)$ are defined in  Theorem \ref {theorem:non-linear_conservation} and Theorem \ref{theo:semi-def} respectively. }} 

\end{rem}

In order to prove Theorem \ref{thr.conv} we use the method developed in  \cite{MouNeu} which (i) relates coercivity estimates on $L$ to the evolution of the corresponding semigroup in the Sobolev space $H^k(\T \times\R^3)$, and (ii) combines spectral gap estimates for the linearized operator with bounds of the nonlinear terms to obtain asymptotic-in-time estimates for the non-linear problem when initial data are sufficiently close to the equilibrium. We summarize the method in the theorem below: 

\begin{theo}  \cite[Thr.~1.1, Thr.~4.1]{MouNeu} \label{MN_Theorem}
\begin{itemize}
\item 
Let $L$ be a linear operator. Assume there exists a suitable decomposition $L = K - \Lambda$ such that 
\begin{align*}
 (i) \; & \nu_{1}\|f\|_{\cH}^{2}\leq (f,\Lambda f)_{L^{2}(\R^{3},dp)}\leq \nu_{2}\|f\|_{\cH}^{2},\\
 (ii)\; &  (\na_{p}\Lambda f,\na_{p} f)_{L^{2}(\R^{3},dp)}\geq \nu_{3}\|\na_{p}f\|_{\cH}^{2} - \nu_{4}\|f\|_{L^{2}_{p}}^{2},\\
 (iii) \; &  (\na_{p}K f, \na_{p}f)_{L^{2}(\R^{3},dp)}\leq C(\delta)\|f\|_{L^{2}(\R^{3},dp)}^{2} + \delta \|\na_{p}f\|_{L^{2}(\R^{3},dp)}^{2},\\ 
 (iv) \; &  |(f,Lg)_{L^2_p}|\leq C\|f\|_\cH\|g\|_\cH ,\\
 (v) \; & -(f,L f)_{L^2(\R^3, dp)} \geq \lambda \|f - \Pi^L f\|_{\cH}^2.
\end{align*}
Then $\mathcal{L}:= L-v\cdot \nabla_x$ generates a strongly continuous evolution semi-group which satisfies
$$
\|e^{\mathcal{L}t} (\mathbb{I} - \Pi^{\mathcal{L}} )\|_{H^1(\T \times\R^3)}\leq C e^{-t/\tau},
$$
for some explicit constants $C$ and $\tau$ that only depend on the constants appearing in $(i)-(v)$. \\

\item Consider the nonlinear problem 
\begin{align}\label{CM_nonlinear}
\partial_t F + v \cdot \nabla_x F = Q(F,F), \quad F(\cdot, 0 ) =  \Fin(\cdot),
\end{align}
and denote by $F^\infty$ the global equilibrium to (\ref{CM_nonlinear}) uniquely determined by the mass, first and second momentum of the initial data. Let 
\begin{align*}
\Gamma(f,f) + Lf :=  \frac{1}{\sqrt{F^\infty}} Q(F_\infty + f\sqrt{F_\infty},F_\infty +f \sqrt{F_\infty}),
 \end{align*}
with $Lf$ a linear operator satisfying  $(v)$ above, and $\Gamma(f,f)$ such that 
\begin{align*}
(ii')\; &  (D_{x}^{\alpha}D_{p}^{\beta}\Lambda f,D_{x}^{\alpha}D_{p}^{\beta}f)_{{{L^{2}(\T\times\R^{3})}}}\geq 
 \nu_{3}\|D_{x}^{\alpha}D_{p}^{\beta}f\|^2_{{{L^{2}(\T,\cH)}}} - \nu_{4}\|f\|^2_{{{H^{k-1}(\T\times\R^{3})}}},\\
 (iii') \; &  (D_{x}^{\alpha}D_{p}^{\beta}K f, D_{x}^{\alpha}D_{p}^{\beta}f)_{L^{2}(\T\times\R^{3})}\leq C(\delta)\|f\|_{H^{k-1}(\T\times\R^{3})}^{2} 
 + \delta \|D_{x}^{\alpha}D_{p}^{\beta}f\|_{L^{2}(\T\times\R^{3})}^{2},\\
 (vi) \; & \| \Gamma(f,f)\|_{H^k(\T \times\R^3)}\le C \| f\|_{H^k(\T \times\R^3)} \left( \sum_{|l|+|j|\le k} \| \partial_x^l\partial_ v^j f \|_{L^{2}(\T, \cH)}\right)^{1/2},
 \end{align*}
 for some $k\ge 4$ and $|\alpha| + |\beta|\le k$, $|\beta|\ge 1$. 
 
 Then (\ref{CM_nonlinear}) has an unique smooth solution that decays exponentially fast towards $F_\infty$:
\begin{align*}
 \left\|\frac{1}{\sqrt{{F^\infty}}}(F - {F^\infty} ) \right\|_{H^k(\T \times\R^3 )} \le  C_{\textrm{in}}\; \varepsilon \; e^{-\lambda t/4}, \qquad t>0,
 \end{align*}
 provided the initial data $\Fin$ satisfies
  \begin{align*}
 \left \|\frac{1}{\sqrt{{F^\infty}}}(\Fin - {F^\infty} ) \right\|_{H^k(\T \times\R^3 )} \le  \varepsilon.
 \end{align*}
 
  \end{itemize}

\end{theo}

Conditions $(i)-(iii)$ state that $\Lambda$ is coercive (in some sense) on the space $\cH$, while $K$ has a regularizing property. 
Assumption $(v)$ is exactly the spectral gap proved in Theorem \ref{thr.spec.gap}. For $(vi)$ will use an estimate proved for the mono-species case by Guo in \cite[Thr.~3]{Guo}.\\

An alternative (and perhaps easier) way of proving Theorem \ref{thr.spec.gap} and Theorem \ref {thr.conv} is to show that $K$ is compact and $\Lambda$ is coercive, see \cite[Lemma 10]{DauJueMouZam}. However this method is non-constructive, in the sense that the size of both the spectral gap and rate of convergence will be only given implicitely. For completeness we add the proof of compactness of $K$ in the Appendix. In the following sections we will adopt the procedure outlined earlier that will  allow for constructive estimates.

 \subsection{Outline} The rest of the paper is organized as follows: after brief summary of the conservation properties for the non-linear system, Section \ref{Conserved quantities and linearization} concerns the formulation of the linearized system and its properties. Section \ref{Proof Teorem 1} contains the proof to Theorem \ref{thr.spec.gap}. In Section \ref{Exponential decay to global equilibrium} we present the proof of Theorem \ref{thr.conv}. Exponential decay is proven with an explicit rate. Finally, in the Appendix we prove the compactness of the operator $K$. 

We conclude by mentioning that among the several  open problems, the one about estimates in the case of very soft potentials $\gamma <-2$ is a particularly interesting question.

    \subsection{Notation} Vectors in $\mathbb{R}^3$ will be denoted by $v,v',p,p'$ and so on, the inner product between $v$ and $w$ will be written $(v,w)$. The identity matrix will be noted by $\mathbb{I}$, the trace of a matrix $X$ will be denoted $\Tr(X)$.  The initial condition for the Cauchy problem will always be denoted by $\fin$ and $C_{{in}}$ will be any positive constant that only depends on the initial data. Unless otherwise specified, $\int dp\equiv \int_{\R^{3}}dp$, $\int dx\equiv \int_{\T}dx$. 
    The space $L^2_p$ denotes the classical Lebesgue spaces $L^2(\R^3)$ with respect to the variable $p$. We denote by $H^k_{x,p}$, $k\ge 1$ the Sobolev space $H^k(\T\times \R^3)$ with respect to the variable $x$ and $p$ and by $L^2_x\mathcal{H}$ the space of all functions with finite norm $\| \|\cdot \|_{\mathcal{H}} \|_{L^2(\T)}$.

\subsection{Acknowledgements.} MPG is supported by NSF DMS-1412748 and DMS-1514761. NZ acknowledges support from the Austrian Science Fund (FWF), grants P24304, P27352, and W1245, and the Austrian-French Program of the Austrian Exchange Service (OeAD).
MPG would like to thank NCTS Mathematics Division Taipei for their kind hospitality. The authors would like to thank Francesco Salvarani for the fruitful discussions. The authors would also like to thank for their kind hospitality the Mathematics Department at Royal Institute of Technology KTH, Sweden.

\section{Conserved quantities and linearization} \label{Conserved quantities and linearization}
In this section we first outline the conservation laws and entropy decay property which hold for (\ref{eq}). Then we present a linearization of   (\ref{eq}) around an equilibrium state and show that the new linear system  also satisfies conservation of mass, total momentum and total energy. 

\begin{theo}\label{theorem:non-linear_conservation}
Let $F_i$, $i = 1,...,N$ be a solution to (\ref{eq})-(\ref{coll_Oper}).
The mass,  the total momentum and energy of the system are conserved over time, i.e. 
\begin{align*}
\frac{d}{dt}\Id\Id    F_i \; dpdx= \frac{d}{dt}  \sum_{i=1}^N \Id \Id  p F_i \; dpdx  =\frac{d}{dt}  \sum_{i=1}^N \Id\Id    \frac{|p|^2}{2 m_{i}} F_i \; dpdx= 0.
\end{align*}
 In addition the Boltzmann entropy functional $H(F_1, F_2, ..., F_N)$ defined as 
\begin{align*}
 H(F_1, F_2, ..., F_N) := \Id\sum_{i=1}^N F_i\log \frac{F_i}{m_{i}^{3}} dp
\end{align*}
decreases along solutions to \eqref{eq}, and it is constant (that is, the entropy production vanishes) if and only if the distribution functions $(F_1,\ldots,F_N)$ are Maxwellians $(\mathcal{M}_1,\ldots,\mathcal{M}_N)$ of the form: 
\begin{align*}
\mathcal{M}_i(x,p) = \frac{{\rho_i}(x)}{\left( 2\pi m_i k_{B}T(x)\right)^{3/2}}e^{-\frac{\left|p-m_i { u(x)}\right|^2}{{2 m_i k_{B} T(x)}}}.
\end{align*}
The density $\rho_i(x)$, velocity $u(x)$ and temperature $T(x)$ are uniquely determined by the conservation properties: 
 $$
 T(x) = \frac{1}{ \sum_1^N \rho_i }  \sum_{i=1}^N\int  \frac{|p - m_{i}u|^2}{3 m_{i} k_{B}} F_i \;dp,  
 \quad u(x) = \frac{1}{ \sum_1^N \rho_i m_{i} }  \sum_{i=1}^N\int p F_i \;dp,
 \quad \rho_i(x) =  \int F_i dp .
 $$
The only local equilibrium that satisfies (\ref{eq})-(\ref{coll_Oper}) is the global equilibrium 
\begin{align*}
\mathcal{M}_i(p) = \frac{{\bar{\rho_i}}}{\left( 2\pi m_i k_{B} T_\infty\right)^{3/2}}e^{-\frac{\left|p-m_i { u_\infty}\right|^2}{{2 m_i k_{B} T_\infty}}},
\end{align*}
 with $\bar{\rho_i}$, $T_\infty$ and $u_\infty$ constants uniquely determined by the conservation properties: 
 $$
 T_\infty = \frac{1}{ \sum_1^N \bar{\rho_i} }  \sum_{i=1}^N\int \int  \frac{|p - m_{i}u|^2}{3m_{i}k_{B}}  F_i \;dpdx,  \;
  u_\infty = \frac{1}{ \sum_1^N  \bar{\rho_i}  m_i }  \sum_{i=1}^N\int \int pF_i \;dpdx,  \;\bar{\rho_i} = \int \int F_i dpdx.
 $$
\end{theo}


\begin{proof}
The mass conservation
follows immediately from the divergence structure of the collision operators.
 We first show total momentum conservation. Integration by parts yields:
\begin{align*}
 &\Id p Q_{ij}(f_i,f_j) dp = -\IId A^{(ij)}\left[\frac{p}{m_i}-\frac{p'}{m_j}\right](f_j'\na f_i - f_i\na f_j')dp dp' \\
 &\quad = \IId f_i f_j' (\Div_p A^{(ij)}\left[\frac{p}{m_i}-\frac{p'}{m_j}\right] - \Div_{p'}A^{(ij)}\left[\frac{p}{m_i}-\frac{p'}{m_j}\right])dp dp'\\
 &\quad = \left( \frac{1}{m_{i}} + \frac{1}{m_{j}} \right)\IId f_i f_j' (\Div_w\Aij[w])\vert_{w=\frac{p}{m_i}-\frac{p'}{m_j}}dp dp' =:   I_{ij}.
\end{align*}
Applying the transformation $p\leftrightarrow p'$ inside $I_{ij}$ and noticing that $w\in\R^{3}\mapsto\Div_{w}A^{(ij)}[w]$ is an odd function, we find that $I_{ij}$ is skew-symmetric: $I_{ij} = -I_{ji}$.\\
Hence, summing up the above equality w.r.t. $i,j=1,\ldots,N$ we get
\begin{align*}
 &\sum_{i,j=1}^N \Id p  Q_{ij}(f_i,f_j) dp =  \sum_{i,j=1}^N I_{ij} =0,
\end{align*}
due to the skew-symmetry of $I_{ij}$. 


Similarly, for the conservation of the total energy, integration by parts yields:
\begin{align*}
 \nonumber
 &\Id \frac{|p|^2}{2} Q_{ij}(f_i,f_j) dp = -\IId p\cdot A^{(ij)}\left[\frac{p}{m_i}-\frac{p'}{m_j}\right](f_j'\na f_i - f_i\na f_j')dp dp'\\
 & = \IId f_i f_j' (\Div_p(A^{(ij)}\left[\frac{p}{m_i}-\frac{p'}{m_j}\right]p) - \Div_{p'}(A^{(ij)}\left[\frac{p}{m_i}-\frac{p'}{m_j}\right]p))dp dp'\\
 & = \IId f_i f_j'\tr(A^{(ij)}\left[\frac{p}{m_i}-\frac{p'}{m_j}\right])dp dp'+\\
 &\qquad + \IId f_i f_j' p\cdot( \Div_p A^{(ij)}\left[\frac{p}{m_i}-\frac{p'}{m_j}\right] - \Div_{p'}A^{(ij)}\left[\frac{p}{m_i}-\frac{p'}{m_j}\right] )dp dp'\\
 & = \IId f_i f_j'\tr(A^{(ij)}\left[\frac{p}{m_i}-\frac{p'}{m_j}\right])dp dp'+\\
 &\qquad + \left(\frac{1}{m_{i}}+\frac{1}{m_{j}}\right)\IId f_i f_j' p\cdot(\Div_w\Aij[w])\vert_{w=\frac{p}{m_{i}}-\frac{p'}{m_{j}}}dp dp'.
\end{align*}
We briefly recall here what we mean when we write $ \Div_w\Aij[w]$. Let $M$ be a $N\times N$ matrix with elements $m_{i,j}$: $\Div_x M$ is a vector with components $b_{i}:=\sum_{j=1}^N\partial_{x_j}m_{i,j}$. Hence 
$$
\Div_z A^{(ij)}[z]  = -2C^{(i,j)}|z|^{\gamma}z.
$$
We denote by $ \Div_x M$ the vector $b$ with components $b_{i}:=\sum_{j=1}^N\partial_{x_j}m_{i,j}$. 
It follows:
\begin{align}
\label{intpQ}
\sum_{i,j=1}^{N}\frac{1}{m_{i}}\Id \frac{|p|^2}{2} Q_{ij}(f_i,f_j) dp = \sum_{i,j=1}^{N}\frac{1}{m_{i}}\IId f_i f_j'\tr(A^{(ij)}\left[\frac{p}{m_i}-\frac{p'}{m_j}\right])dp dp'& \\
+ \sum_{i,j=1}^{N}\left(\frac{1}{m_{i}}+\frac{1}{m_{j}}\right)\IId f_i f_j' \frac{p}{m_{i}}\cdot(\Div_w\Aij[w])\vert_{w=\frac{p}{m_{i}}-\frac{p'}{m_{j}}}&dp dp' .
\nonumber
\end{align}
By applying the transformation $(p,i)\leftrightarrow (p',j)$ in the terms on the right-hand side of \eqref{intpQ} we deduce:
\begin{align*}
& \sum_{i,j=1}^{N}\frac{1}{m_{i}}\Id \frac{|p|^2}{2} Q_{ij}(f_i,f_j) dp
= \frac{1}{2}\sum_{i,j=1}^{N}\left(\frac{1}{m_{i}}+\frac{1}{m_{j}}\right)\IId f_i f_j'\tr(A^{(ij)}\left[\frac{p}{m_i}-\frac{p'}{m_j}\right])dp dp' + \\
&\qquad + \frac{1}{2}\sum_{i,j=1}^{N}\left(\frac{1}{m_{i}}+\frac{1}{m_{j}}\right)\IId f_i f_j' 
(w\cdot\Div_w\Aij[w])\vert_{w=\frac{p}{m_{i}}-\frac{p'}{m_{j}}}dp dp' =0,\nonumber
\end{align*}
since $w\cdot\Div_w\Aij[w] = -\tr\Aij[w]$ for $w\in\R^{3}$. The total energy conservation follows.

Finally, we show that the entropy functional $H$ is decreasing as time increases: 
\begin{align*}
 -\frac{d }{dt}H(f_1, f_2, ...,f_N) 
 &= -\sum_{i,j=1}^N\Id (\log f_i + 1)Q_{ij}(f_i,f_j)dp \\
 \nonumber
 & = \sum_{i,j=1}^N\IId \frac{\na f_i}{f_i}\cdot A^{(ij)}\left[\frac{p}{m_i}-\frac{p'}{m_j}\right](f_j'\na f_i - f_i\na f_j')dp dp'\\
 \nonumber
 & = \sum_{i,j=1}^N\IId f_i f_j' \frac{\na f_i}{f_i}\cdot A^{(ij)}\left[\frac{p}{m_i}-\frac{p'}{m_j}\right]\left( \frac{\na f_i}{f_i} - \frac{\na f_j'}{f_j'} \right)dp dp'.
\end{align*}
By exchanging $i\leftrightarrow j$ and $p\leftrightarrow p'$ we obtain:
\begin{align*}
  -\frac{d }{dt}H 
 &=\sum_{i,j=1}^N\IId f_i f_j' \frac{\na f_i}{f_i}\cdot A^{(ij)}\left[\frac{p}{m_i}-\frac{p'}{m_j}\right]\left( \frac{\na f_i}{f_i} - \frac{\na f_j'}{f_j'} \right)dp dp'\\
 \nonumber
 &=-\sum_{i,j=1}^N\IId f_i f_j' \frac{\na f_j'}{f_j'}\cdot A^{(ij)}\left[\frac{p}{m_i}-\frac{p'}{m_j}\right]\left( \frac{\na f_i}{f_i} - \frac{\na f_j'}{f_j'} \right)dp dp'\\
 \nonumber
 &=\frac{1}{2}\sum_{i,j=1}^N\IId f_i f_j' \left(\frac{\na f_i}{f_i}-\frac{\na f_j'}{f_j'}\right)
 \cdot A^{(ij)}\left[\frac{p}{m_i}-\frac{p'}{m_j}\right]\left( \frac{\na f_i}{f_i} - \frac{\na f_j'}{f_j'} \right)dp dp'\geq 0,
\end{align*}
since $A^{(ij)}$ is a positive definite matrix. 

 Hence, $\frac{d }{dt}H=0$ if and only if 
$\frac{\na f_i}{f_i}-\frac{\na f_j'}{f_j'}$ lies in the kernel of $A^{(ij)}\left[\frac{p}{m_i}-\frac{p'}{m_j}\right]$, that is, if and only if there exists a scalar function $\lambda_{ij}[v,v']: \R^3 \times \R^3 \to \R$ such that 
\begin{align}\label{stat}
 \frac{\na f_i}{f_i}-\frac{\na f_j'}{f_j'} = \lambda_{ij}\left[\frac{p}{m_{i}},\frac{p'}{m_{j}}\right]\left(\frac{p}{m_{i}}-\frac{p'}{m_{j}}\right).
\end{align}
We next show that the matrix $\{ \lambda_{ij}[\frac{p}{m_{i}},\frac{p}{m_{i}}]\}_{i,j}$ is constant for all $i$ and $j$.  Applying the transformation $(p,i)\leftrightarrow(p',j)$ in \eqref{stat} we get 
$$\lambda_{ij}\left[\frac{p}{m_{i}},\frac{p'}{m_{j}}\right] = \lambda_{ji}\left[\frac{p'}{m_{j}},\frac{p}{m_{i}}\right],$$
 which implies
\begin{align*}
& \lambda_{ij}\left[\frac{p}{m_{i}},\frac{p}{m_{i}}\right] = \lambda_{ji}\left[\frac{p}{m_{i}},\frac{p}{m_{i}}\right]. 
\end{align*}
We differentiate \eqref{stat} w.r.t. $p$ and obtain:
\begin{align*}
 D^2\log f_i(p) = \na_p\lambda_{ij}\left[\frac{p}{m_{i}},\frac{p'}{m_{j}}\right]\otimes\left(\frac{p}{m_{i}}-\frac{p'}{m_{j}}\right) 
 + \frac{1}{m_{i}}\lambda_{ij}\left[\frac{p}{m_{i}},\frac{p'}{m_{j}}\right] {\mathbb{I}}.
\end{align*}
Consequently for $p'/m_{j}=p/m_{i}$,
\begin{align}\label{stat.2}
\pa_{p_k p_s}^2\log f_i(p) = \frac{1}{m_{i}}\lambda_{ij}\left[\frac{p}{m_{i}},\frac{p}{m_{i}}\right]\delta_{ks},\qquad k,s=1,2,3.
\end{align}
Differentiation of \eqref{stat.2} leads to:
\begin{align*}
\pa_{p_\ell}\pa_{p_k p_s}^2\log f_i(p) =\frac{1}{m_{i}}\pa_{p_\ell} \lambda_{ij}\left[\frac{p}{m_{i}},\frac{p}{m_{i}}\right]\delta_{ks},\qquad k,s,\ell=1,2,3.
\end{align*}
Since the order of the derivatives on the left hand side is interchangeable (assuming enough smoothness for $f_i$), one deduces that
\begin{align*}
\pa_{p_\ell}\lambda_{ij}\left[\frac{p}{m_{i}},\frac{p}{m_{i}}\right]\delta_{ks} 
= \pa_{p_k}\lambda_{ij}\left[\frac{p}{m_{i}},\frac{p}{m_{i}}\right]\delta_{\ell s} ,
\qquad k,s,\ell = 1,2,3,
\end{align*}
which is consistent if and only if $v\in\R^{3}\mapsto\lambda_{ij}[v,v]$ is constant.
 
Moreover, \eqref{stat.2} implies that, for $i=1,\ldots,N$, $\lambda_{ij}$ does not depend on $j$. Summarizing, we have found that 
$\lambda_{i,j}[v,v]$ is constant, symmetric in $i,j$ and does not depend on $j$. 
Hence $\lambda_{i,j}[v,v] \equiv -\alpha^{(2)}$, $ \alpha^{(2)}\in \R$, for $i,j=1,\ldots,N$, $v\in\R^{3}$.
This fact and \eqref{stat.2} imply that $\log f_i(p)$ is a second order polynomial in $p$:
\begin{align}
 \log f_i(p) = \alpha_i^{(0)} + \alpha_i^{(1)}\cdot p - \alpha^{(2)} \frac{|p|^2}{2 m_{i}} , \qquad i=1,\ldots,N.\label{stat.3}
\end{align}
From \eqref{stat} and \eqref{stat.3} it follows:
\begin{align*}
\alpha_i^{(1)} - \alpha_j^{(1)} - \alpha^{(2)} \left(\frac{p}{m_{i}}  - \frac{p'}{m_{j}} \right)
= -\alpha^{(2)}\left(\frac{p}{m_{i}}-\frac{p'}{m_{j}}\right),
\end{align*}
which leads to $\alpha_i^{(1)} = \alpha_j^{(1)}$, $i,j=1,\ldots,N$ after evaluation for $p'/m_{j}=p/m_{i}$. 
We conclude that 
\begin{align*}
& \log f_i(p) = \alpha_i^{(0)} + \alpha^{(1)}\cdot p - \alpha^{(2)}\frac{|p|^2}{2m_{i}},\qquad \quad  ~�~�i=1,\ldots,N.
\end{align*}
Conservation of mass, momentum and energy uniquely determine the constants $\alpha_i^{(0)} $, $\alpha^{(1)}$ and $\alpha^{(2)}$.

\end{proof}

\subsection*{Linearization around the equilibrium.} 

We now linearize the collision operator $Q$ around the Maxwellians $(M_1,\ldots,M_N)$ defined as 
$$ M_{i}(p) = \frac{\rho_{i}}{(2\pi m_{i}k_{B}T)^{3/2}}e^{-\frac{1}{2}\frac{|p|^{2}}{ m_{i}k_{B}T }} .$$
 It holds:
\begin{align*}
 \sum_{j=1}^N Q_{ij}(M_i+\sqrt{M_i}f_i,M_j+\sqrt{M_j}f_j) =& \sum_{j=1}^N  Q_{ij}(M_i,\sqrt{M_j}f_j) + Q_{ij}(\sqrt{M_i}f_i,M_j) \;+\\
 & + \;Q_{ij}(\sqrt{M_i}f_i,\sqrt{M_j}f_j), 
\end{align*}
taking into account that $Q_{i,j}(M_i, M_j)=0$. 
Let us first compute:
\begin{align*}
  Q_{ij}(\sqrt{M_i}f_i,M_j)& = \Div_p\Id A^{(ij)}\left[\frac{p}{m_i}-\frac{p'}{m_j}\right](M_j'\na(\sqrt{M_i}f_i) - \sqrt{M_i}f_i\na M_j')dp'\\
 &= \Div_p\Id A^{(ij)}\left[\frac{p}{m_i}-\frac{p'}{m_j}\right]\left((M_j'\na\sqrt{M_i} - \sqrt{M_i}\na M_j')f_i + M_j'\sqrt{M_i}\na f_i \right)dp'.
\end{align*}
Rewriting 
\begin{align}\label{computations}
 M_j'\na\sqrt{M_i} - \sqrt{M_i}\na M_j' 
&= \sqrt{M_i}M_j'\left( -\frac{1}{2}\na\log M_j' -\frac{1}{2 k_{B}T}\left(\frac{p}{m_i}-\frac{p'}{m_j}\right) \right) \nonumber\\
&= -\sqrt{M_i M_j'}\na \sqrt{M_j'} -\frac{\sqrt{M_i}M_j'}{2 k_{B}T}\left(\frac{p}{m_i}-\frac{p'}{m_j}\right),
\end{align}
it follows:
\begin{align}
 Q_{ij}(\sqrt{M_i}f_i,M_j) &= \Div_p\Id\sqrt{M_i M_j'}A^{(ij)}\left[\frac{p}{m_i}-\frac{p'}{m_j}\right]\left( \sqrt{M_j'}\na f_i - f_i\na\sqrt{M_j'} \right) dp' , \label{Lin.1}
\end{align}
since $A^{(ij)}\left[\frac{p}{m_i}-\frac{p'}{m_j}\right](\frac{p}{m_i}-\frac{p'}{m_j})\equiv 0$.
We now consider
\begin{align*}
 Q_{ij}(M_i,\sqrt{M_j}f_j) &= \Div_p\Id A^{(ij)}\left[\frac{p}{m_i}-\frac{p'}{m_j}\right]\left( \sqrt{M_j'}f_j'\na M_i - M_i\na\left(\sqrt{M_j'}f_j'\right) \right) dp'\\
 &= \Div_p\Id A^{(ij)}\left[\frac{p}{m_i}-\frac{p'}{m_j}\right]\left( f_j'\left(\sqrt{M_j'}\na M_i - M_i\na\sqrt{M_j'}\right) - M_i\sqrt{M_j'}\na f_j' \right)dp'.
\end{align*}
Using similar calculations as in (\ref {computations}) one gets
$$
\sqrt{M_j'}\na M_i - M_i\na\sqrt{M_j'} = \sqrt{M_i M_j'}\na\sqrt{M_i} + \frac{M_i\sqrt{M_j'}}{2 k_{B}T}\left(\frac{p}{m_i}-\frac{p'}{m_j}\right),
$$
which implies
\begin{align}
 Q_{ij}(M_i,\sqrt{M_j}f_j) = \Div_p\Id\sqrt{M_i M_j'}A^{(ij)}\left[\frac{p}{m_i}-\frac{p'}{m_j}\right]\left( f_j'\na\sqrt{M_i} - \sqrt{M_i}\na f_j' \right)dp'.\label{Lin.2}
\end{align}
Adding \eqref{Lin.1} with \eqref{Lin.2} (and dividing by $\sqrt{M_{i}}$) we obtain the linearized collision operator:
\begin{align*}
 L_i(f_1,\ldots,f_n) &= \sum_{j=1}^N L_{ij}(f_i,f_j),
 \end{align*}
 with
 \begin{align}
 \label{L.2}
 L_{ij}(f_i,f_j) &:= \frac{1}{\sqrt{M_{i}}}\Div_p\Id\sqrt{M_i M_j'}A^{(ij)}\left[\frac{p}{m_i}-\frac{p'}{m_j}\right]\cdot \\
 &\qquad \cdot \Big(\sqrt{M_j'}\na f_i - \sqrt{M_i}\na f_j' - f_i\na\sqrt{M_j'} + f_j'\na\sqrt{M_i} \Big) dp' .\nonumber
\end{align}
We briefly recall the conserved quantities for $L_i$:
\begin{theo}\label{theorem:linear_conservation}
Let $f_i$, $i=1,...,N$ be the solution to the linear system: 
\begin{equation*}
 \left\{ \begin{array}{c} 
                         \pa_t f_i + \frac{p}{m_i} \cdot \nabla_x f_i =   \sum_{j=1}^N L_{ij}(f_i,f_j), \\ 
                         {}\\
                         f(x,p,0) = \fin(x,p),\\
           \end{array}  \right.
\end{equation*}
with $L_{ij}$ defined as in (\ref{L.2}). The mass $\int \Id \sqrt{M_{i}}f_i \; dpdx$, total momentum $\sum_{i=1}^N \int \Id p \sqrt{M_{i}}f_i \; dpdx$ and total energy 
$\sum_{i=1}^N \int \Id (|p|^2/2m_i) \sqrt{M_{i}}f_i \; dpdx$ are constant in time.

\end{theo}

\begin{proof}
 The mass of each function $\sqrt{M_{i}} f_i$ is conserved because of the divergence form of the operator. Moreover, with an integration by parts we can deduce
\begin{align*}
 \Id p\sum_{i=1}^N \sqrt{M_{i}} &L_i(f_1,\ldots,f_N) dp\\
 = -\sum_{i,j=1}^N\IId &
 \sqrt{M_i M_j'}A^{(ij)}\left[\frac{p}{m_i}-\frac{p'}{m_j}\right]\cdot \\
 &\cdot\left( \sqrt{M_j'}\na f_i - \sqrt{M_i}\na f_j' - f_i\na\sqrt{M_j'} + f_j'\na\sqrt{M_i} \right) dp dp' = 0,
\end{align*}
because the quantity inside the integral is antisymmetric for the transformation $(i,p)\leftrightarrow (j,p')$. Finally, the same transformation
and another integration by parts allow us to write:
\begin{align*}
 \Id\sum_{i=1}^N &\frac{|p|^2}{2m_{i}}\sqrt{M_{i}}L_i(f_1,\ldots,f_N) dp\\
 =  -\sum_{i,j=1}^N&\IId
 \sqrt{M_i M_j'}\frac{p}{m_{i}}\cdot A^{(ij)}\left[\frac{p}{m_i}-\frac{p'}{m_j}\right]\cdot \\
 & \cdot \left( \sqrt{M_j'}\na f_i - \sqrt{M_i}\na f_j' - f_i\na\sqrt{M_j'} + f_j'\na\sqrt{M_i} \right) dp dp'\\
 = -\sum_{i,j=1}^N&\IId
\frac{1}{2} \sqrt{M_i M_j'}\left(\frac{p}{m_{i}}-\frac{p'}{m_{j}}\right)\cdot A^{(ij)}\left[\frac{p}{m_i}-\frac{p'}{m_j}\right]\cdot \\
 &\cdot \left( \sqrt{M_j'}\na f_i - \sqrt{M_i}\na f_j' - f_i\na\sqrt{M_j'} + f_j'\na\sqrt{M_i} \right) dp dp' = 0,
\end{align*}
since $A^{(ij)}\left[\frac{p}{m_i}-\frac{p'}{m_j}\right](\frac{p}{m_{i}}-\frac{p'}{m_{j}})=0$. The proof is complete. 
\end{proof}

\subsection*{Structure of the linearized collision operator.}
We first show that $L_{ij}$ can be rewritten in the following form:
\begin{align}
 L_{ij}(f_i,f_j) = \frac{1}{\sqrt{M_i}}\Div_p\int M_i M_j' A^{(ij)}\left[ \frac{p}{m_i} - \frac{p'}{m_j} \right]
 \left( \na\left( \frac{f_i}{\sqrt{M_i}} \right) - \na\left( \frac{f_j'}{\sqrt{M_j'}} \right) \right) dp' .  \label{Lij.b}
\end{align}
To prove \eqref{Lij.b} we first notice that:
\begin{align*}
 \na\log\sqrt{M_j'} = -\frac{1}{2 k_{B}T} \frac{p'}{m_j} = \frac{1}{2k_{B}T}\left( \frac{p}{m_i}-\frac{p'}{m_j}\right) +\na\log\sqrt{M_i}.
\end{align*}
It follows that the term $\sqrt{M_j'}\na f_i - f_i\na\sqrt{M_j'}$ inside \eqref{L.2} can be rewritten as:
\begin{align*}
 \sqrt{M_j'}\na f_i - f_i\na\sqrt{M_j'} &= \sqrt{M_i M_j'}\left( \frac{\na f_i}{\sqrt{M_i}} - \frac{f_i}{\sqrt{M_i}}\na\log\sqrt{M_j'} \right)\\
& = \sqrt{M_i M_j'}\left( \frac{\na f_i}{\sqrt{M_i}} - \frac{f_i}{\sqrt{M_i}}\na\log\sqrt{M_i} \right) - \frac{f_{i}\sqrt{M_{j}'}}{2k_{B}T}\left( \frac{p}{m_i}-\frac{p'}{m_j}\right)\\
& = \sqrt{M_i M_j'}\na\left( \frac{f_i}{\sqrt{M_i}} \right) - \frac{f_{i}\sqrt{M_{j}'}}{2k_{B}T}\left( \frac{p}{m_i}-\frac{p'}{m_j}\right).
\end{align*}
The other term $\sqrt{M_i}\na f_j' - f_j'\na\sqrt{M_i}$ is treated in a similar way. This shows that \eqref{Lij.b} and \eqref{L.2}  are equivalent formulations.\\

{{We will now decompose the operator $L = (L_1, L_2, ..., L_N)$ as $L = L^m + L^b$, where $L^m$ and $L^b$ respectively describe collisions between particles of the same species and of different species. More precisely,
\begin{align*}
 L^m(f) &:=  (L_{11}(f_1,f_1), ...,L_{NN}(f_N,f_N)), \\
 L^b(f) &:= (\sum_{j\neq 1}L_{1j}(f_1,f_j), ..., \sum_{j\neq N} L_{Nj}(f_N,f_j).
\end{align*}

\begin{theo} \label{theo:semi-def}
Both operators $L^m$ and $L^b$ are negative semidefinite. Moreover $f \in N(L)$ if and only if 
\begin{align*}
 f_i = M_i^{1/2}\left( \beta^{(0)}_i + \beta^{(1)}\cdot p + \beta^{(2)} \frac{|p|^2}{2 m_i} \right),
 \qquad i=1,\ldots,N , 
\end{align*}
for some $p$-independent real coefficients $\beta^{(0)}_i$, $i=1,\ldots,N$, $\beta^{(1)}$ and $\beta^{(2)} $, and $f \in N(L^m)$ if and only if:
\begin{align*}
 f_i = M_i^{1/2}\left( \alpha^{(0)}_i + \alpha^{(1)}_i\cdot p + \alpha^{(2)}_i |p|^2 \right),
 \qquad i=1,\ldots,N ,
\end{align*}
for some $p-$independent real coefficients $\alpha^{(0)}_i$, $\alpha^{(1)}_i$, $\alpha^{(2)}_i $, $i=1,\ldots,N$.

\end{theo}
\begin{proof}

A change of variable $p\leftrightarrow p'$ allows to write
 \begin{align*}
 (f,L^m f)_{L_p^2} := &\sum_{i=1}^N (f_i,L_{ii}( f_i,f_i))_{L_p^2} \\
=&  -\frac{1}{2} \sum_{i=1}^N\Id\Id   M_i M_i' A^{(ii)}\left[ \frac{p}{m_i} - \frac{p'}{m_i} \right]
 \left( \na\left( \frac{f_i}{\sqrt{M_i}} \right) - \na\left( \frac{f_i'}{\sqrt{M_i'}} \right) \right)\cdot  \\
 & \qquad \qquad \cdot  \left( \na\left( \frac{f_i}{\sqrt{M_i}} \right) - \na\left( \frac{f_i'}{\sqrt{M_i'}} \right) \right) \;dpdp' \le 0.
 \end{align*}
 Using the same change of variable, for each $i\neq j$ one can show that 
 \begin{align}
   (f_i,L_{ij}( f_i,f_j))_{L_p^2} +&  (f_j,L_{ji}( f_j,f_i))_{L_p^2} \nonumber \\
   = -& \Id\Id   M_i M_j' A^{(ii)}\left[ \frac{p}{m_i} - \frac{p'}{m_j} \right]
 \left( \na\left( \frac{f_i}{\sqrt{M_i}} \right) - \na\left( \frac{f_j'}{\sqrt{M_j'}} \right) \right)\cdot  \nonumber\\
 & \qquad \qquad \cdot  \left( \na\left( \frac{f_i}{\sqrt{M_i}} \right) - \na\left( \frac{f_j'}{\sqrt{M_j'}} \right) \right) \;dpdp' \le 0, \label{skew_symm}
 \end{align}
  which yields $(f,L^b f)_{L_p^2} := \sum_{\substack{i,j=1\\ j\neq i}}^N (f_i, L_{ij}(f_i,f_j))_{L_p^2}  \leq 0$ for all $f\in D(L)$.

It is clear that $(f,L^m f)_{L_p^2}=0$ if and only if 
\begin{align*}
 \na\left( \frac{f_i}{\sqrt{M_i}} \right) - \na\left( \frac{f_i'}{\sqrt{M_i'}} \right) 
 = \mu_{ij}[p,p']\left( \frac{p}{m_i} - \frac{p'}{m_i} \right).
\end{align*}
By employing the same method that was used to solve \eqref{stat} we find that $f \in N(L^m)$ if and only if:
\begin{align}
 f_i = M_i^{1/2}\left( \alpha^{(0)}_i + \alpha^{(1)}_i\cdot p + \alpha^{(2)}_i |p|^2 \right),
 \qquad i=1,\ldots,N , \label{N.Lm}
\end{align}
for some $p-$independent real coefficients $\alpha^{(0)}_i$, $\alpha^{(1)}_i$, $\alpha^{(2)}_i$, $i=1,\ldots,N$.
Eq.~\eqref{N.Lm} is a complete characterization of $N(L^m)$. A similar strategy yields the description of the kernel of $L$: $f \in N(L)$ if and only if 
\begin{align}
 f_i = M_i^{1/2}\left( \beta^{(0)}_i + \beta^{(1)}\cdot p + \beta^{(2)} \frac{|p|^2}{2 m_i} \right),
 \qquad i=1,\ldots,N , \label{N.L}
\end{align}
for some $p$-independent real coefficients $\beta^{(0)}_i$, $i=1,\ldots,N$, $\beta^{(1)}$ and $\beta^{(2)} $. Eq.~\eqref{N.L} is a complete characterization of $N(L)$.

\end{proof}
}}

\section{Proof of Theorem  ~\ref{thr.spec.gap}}\label{Proof Teorem 1}
This section is devoted to the proof of Theorem  ~\ref{thr.spec.gap} which states that the multi-species linearized Landau collision operator $L=(L_1,L_2,...,L_N)$ defined as in (\ref{L.2}) has a spectral gap in the Hilbert space $\cH$.\\

 The starting point in the proof is the already known spectral gap for the mono-species operator proven in several works, including \cite{Guo,MS07} and summarized in the next lemma.
\begin{lemma}\label{lem.spec.gap}
There exists an explicitly computable constant $\lambda_m>0$ such that:
\begin{align*}
 -(f,L^m f) \geq \lambda_m\|f - \Pi^m f\|_{\cH}^2\qquad f\in D(L^m),
\end{align*} 
where $\Pi^m$ denotes the projection operator onto the subspace $N(L^m)$.
\end{lemma}

We will now follow an approach similar to the one formulated in \cite{DauJueMouZam}. 
We first write 
$$
f = f^\parallel  + f^\perp,
$$
with
$$
f^\parallel := \Pi^m f, \quad f^\perp :=  (\mathbb{I}-\Pi^m)f .
$$
From (\ref{skew_symm}) it follows:
\begin{align*}
 -(f,L^b f)_{L^2_p} = \frac{1}{2}\sum_{\substack{i,j=1\\ j\neq i}}^N\iint M_i M_j' (w_p + w_o)\cdot 
 A^{(ij)}\left[\frac{p}{m_i}-\frac{p'}{m_j}\right] (w_p + w_o) dp dp',
\end{align*}
with 
\begin{align*}
  w_p := \na\left( \frac{f_i^\parallel}{\sqrt{M_i}} \right) - \na\left( \frac{(f_j^\parallel)'}{\sqrt{M_j'}} \right),\qquad
 w_o := \na\left( \frac{f_i^\perp}{\sqrt{M_i}} \right) - \na\left( \frac{(f_j^\perp)'}{\sqrt{M_j'}} \right).
\end{align*}
Since $A^{(ij)}$ is symmetric and positive definite, Young's inequality yields
$$
\frac{1}{2} w_p \cdot A^{(ij)} w_o + \frac{1}{2} w_o \cdot A^{(ij)} w_p = w_p \cdot A^{(ij)} w_o \ge -\frac{1}{4} w_p \cdot A^{(ij)} w_p -  w_o \cdot A^{(ij)} w_o, 
$$
and 
\begin{align}
 -(f,L^b f)_{L^2_p} \geq & \;\frac{1}{4}
 \sum_{\substack{i,j=1\\ j\neq i}}^N\iint M_i M_j' w_p \cdot A^{(ij)}\left[\frac{p}{m_i}-\frac{p'}{m_j}\right] w_p dp dp' \nonumber \\
 & -\frac{1}{2}\sum_{\substack{i,j=1\\ j\neq i}}^N\iint M_i M_j' w_o \cdot A^{(ij)}\left[\frac{p}{m_i}-\frac{p'}{m_j}\right] w_o dp dp' \nonumber \\
 =& -\frac{1}{2}(f^\parallel,L^b f^\parallel)_{L^2_p} -\frac{1}{2}\sum_{\substack{i,j=1\\ j\neq i}}^N\iint M_i M_j' w_o \cdot
 A^{(ij)}\left[\frac{p}{m_i}-\frac{p'}{m_j}\right] w_o dp dp'.\label{fLbf}
\end{align}
Let us estimate the second term on the right-hand side of \eqref{fLbf}. Applying Young's inequality one more time we get  
\begin{align*}
  \frac{1}{2}\sum_{\substack{i,j=1\\ j\neq i}}^N&\iint M_i M_j' w_o \cdot A^{(ij)}\left[\frac{p}{m_i}-\frac{p'}{m_j}\right] w_o dp dp'\\
 & \leq\sum_{\substack{i,j=1\\ j\neq i}}^N\iint M_i M_j' \na\left( \frac{f_i^\perp}{\sqrt{M_i}} \right) \cdot
  A^{(ij)}\left[\frac{p}{m_i}-\frac{p'}{m_j}\right]\na\left( \frac{f_i^\perp}{\sqrt{M_i}} \right) dp dp'\\
 &\qquad + \sum_{\substack{i,j=1\\ j\neq i}}^N\iint M_i M_j' \na\left( \frac{(f_j^\perp)'}{\sqrt{M_j'}}\right) \cdot
  A^{(ij)}\left[\frac{p}{m_i}-\frac{p'}{m_j}\right]\na\left( \frac{(f_j^\perp)'}{\sqrt{M_j'}}\right) dp dp'\\
 & = 2\sum_{\substack{i,j=1\\ j\neq i}}^N\iint M_i M_j' \na\left( \frac{f_i^\perp}{\sqrt{M_i}} \right) \cdot
  A^{(ij)}\left[\frac{p}{m_i}-\frac{p'}{m_j}\right]\na\left( \frac{f_i^\perp}{\sqrt{M_i}} \right) dp dp'.
\end{align*}
Since
\begin{align*}
 \na\left(\frac{f_i^\perp}{\sqrt{M_i}}\right) 
 &= \frac{\na f_i^\perp}{\sqrt{M_i}} - \frac{f_i^\perp}{\sqrt{M_i}}\na\log\sqrt{M_i},
\end{align*}
we have
\begin{align}
 \frac{1}{2}\sum_{\substack{i,j=1\\ j\neq i}}^N&\iint M_i M_j' w_o \cdot A^{(ij)}\left[\frac{p}{m_i}-\frac{p'}{m_j}\right] w_o dp dp' \nonumber\\
 \nonumber
 &\leq 4\sum_{\substack{i,j=1\\ j\neq i}}^N\iint M_j' 
 \na f_i^\perp\cdot A^{(ij)}\left[\frac{p}{m_i}-\frac{p'}{m_j}\right]\na f_i^\perp dp dp'\\
 \nonumber
 &\qquad + 4\sum_{\substack{i,j=1\\ j\neq i}}^N\iint M_j' 
 (f_i^\perp)^2 \na\log\sqrt{M_i}\cdot A^{(ij)}\left[\frac{p}{m_i}-\frac{p'}{m_j}\right]
 \na\log\sqrt{M_i} dp dp'\\
 &\leq \sum_{i=1}^N\int \na f_i^\perp\cdot\cA^{(i)}\na f_i^\perp dp
 + \sum_{i=1}^N\int (f_i^\perp)^2 \cB^{(i)} dp,\label{intwo.1}
\end{align}
with
\begin{align*}
 \cA^{(i)} &:= 4\sum_{j=1}^N\int M_j' A^{(ij)}\left[\frac{p}{m_i}-\frac{p'}{m_j}\right] dp',\qquad
 \cB^{(i)} := \na\log \sqrt{M_i}\cdot\cA^{(i)}\na\log\sqrt{M_i}.
\end{align*}
From \cite[Lemma 2.3]{CTW} we deduce that:
\begin{align}
 \label{est.cA}
 \na f_i^\perp\cdot\cA^{(i)}\na f_i^\perp
 &\leq C\left( \bk{p}^{\gamma}|P\na f_i^\perp|^2 + \bk{p}^{\gamma+2}|(I-P)\na f_i^\perp|^2 \right),\\
 \cB^{(i)} &\leq C\bk{p}^{\gamma+2} . \label{est.cB}
\end{align}
Inequalities \eqref{intwo.1}, \eqref{est.cA} and \eqref{est.cB} imply:
\begin{align}\label{intwo}
 \frac{1}{2}\sum_{\substack{i,j=1\\ j\neq i}}^N\iint M_i M_j' w_o \cdot A^{(ij)}\left[\frac{p}{m_i}-\frac{p'}{m_j}\right] w_o dp dp'
 \leq C_1\|f^\perp\|_\cH^2 ,
\end{align}
for some explicitly computable constant $C_1>0$. In summary we have shown that 
\begin{align} \label{Lb_est}
-(f,L^b f)_{L^2_p} \geq-\frac{1}{2}(f^\parallel,L^b f^\parallel)_{L^2_p} - C_1\|f^\perp\|_\cH^2 .
\end{align}

We are now ready to prove the next lemma:

\begin{lemma}\label{lem.fLf.1}
For each $f\in D(L)$ and $\eta\in (0,1]$ we have
\begin{align*}
 -(f,L f)_{L^2_p}\geq (\lambda_m - \eta C_1)\|f^\perp\|_\cH^2 -\frac{\eta}{2}(f^\parallel,L^b f^\parallel)_{L^2_p}.
\end{align*}
\end{lemma}

\begin{proof}
Using the decomposition $L = L^m + L^b$ we get, 
\begin{align*}
-(f,L f)_{L^2_p} &= -(f,L^m f)_{L^2_p} -(f,L^b f)_{L^2_p}\\
& \ge  -(f,L^m f)_{L^2_p} -\eta (f,L^b f)_{L^2_p}
\end{align*}
for each $\eta\in (0,1]$, since $L^b$ is a negative semidefinite operator, as shown in Theorem \ref{theo:semi-def}.
Finally Lemma \ref{lem.spec.gap} and (\ref{Lb_est}) imply 
\begin{align*}
-(f,L f)_{L^2_p} & \ge  \lambda_m\|f ^\perp\|_{\cH}^2 -\eta\left( \frac{1}{2}(f^\parallel,L^b f^\parallel)_{L^2_p}  +  C_1\|f^\perp\|_\cH^2\right),
\end{align*}
 which finishes the proof.

\end{proof}

We focus now our attention on $(f^\parallel,L^b f^\parallel)_{L^2_p}$. From \eqref{N.Lm} it follows 
\begin{align}
 f_i^\parallel = (\Pi^m f)_i = M_i^{1/2}\left(\alpha_i + u_i\cdot p + e_i\frac{|p|^2}{2 m_i}\right),\qquad i=1,\ldots,N, \label{f.par}
\end{align}
for a suitable choice of $\alpha_i$, $u_i$, $e_i$. We get:
\begin{align*}
 &-(f^\parallel,L^b f^\parallel)_{L^2_p}
 = \frac{1}{2}\sum_{\substack{i,j=1\\ j\neq i}}^N\iint M_i M_j'
 \left( u_i - u_j + e_i\frac{p}{m_i} - e_j\frac{p'}{m_j} \right) \cdot\\
 &\qquad\qquad \cdot A^{(ij)}\left[\frac{p}{m_i}-\frac{p'}{m_j}\right]
 \left( u_i - u_j + e_i\frac{p}{m_i} - e_j\frac{p'}{m_j} \right)dp dp' .
\end{align*}
We first notice that 
\begin{align*}
 \left( u_i - u_j \right)  &\cdot \iint M_i M_j'
 A^{(ij)}\left[\frac{p}{m_i}-\frac{p'}{m_j}\right]
 \left(  e_i\frac{p}{m_i} - e_j\frac{p'}{m_j} \right)dp dp'\\
 =& \left( u_i - u_j \right)  \cdot  \frac{e_i}{m_i} \Id M_i p \left( \Id  M_j'
 A^{(ij)}\left[\frac{p}{m_i}-\frac{p'}{m_j}\right] dp' \right)dp\\
 &-  \left( u_i - u_j \right)  \cdot  \frac{e_j}{m_j} \Id M_j' p' \left(\Id  M_i
 A^{(ij)}\left[\frac{p}{m_i}-\frac{p'}{m_j}\right] dp\right)dp' .
 \end{align*}
Since the function 
$(p,p')\in\R^3\times\R^3\mapsto M_i M_j'
 A^{(ij)}\left[\frac{p}{m_i}-\frac{p'}{m_j}\right]
 \left(  e_i\frac{p}{m_i} - e_j\frac{p'}{m_j} \right)\in\R^3$ is odd,
it follows that:
 $$
  \left( u_i - u_j \right)  \cdot \iint M_i M_j'
 A^{(ij)}\left[\frac{p}{m_i}-\frac{p'}{m_j}\right]
 \left(  e_i\frac{p}{m_i} - e_j\frac{p'}{m_j} \right)dp dp' =0.
 $$
 Hence we are left with 
 \begin{align*}
 &-(f^\parallel,L^b f^\parallel)_{L^2_p}
 = \sum_{i,j=1}^N (u_i-u_j)\cdot\iint M_i M_j' A^{(ij)}\left[\frac{p}{m_i}-\frac{p'}{m_j}\right]dp dp'~(u_i-u_j)\\
 & + \sum_{i,j=1}^N \frac{(e_i-e_j)^2}{4}\iint M_i M_j' \left(\frac{p}{m_i}+\frac{p'}{m_j}\right)\cdot
 A^{(ij)}\left[\frac{p}{m_i}-\frac{p'}{m_j}\right]\left(\frac{p}{m_i}+\frac{p'}{m_j}\right) dp dp' ,\nonumber
\end{align*}
after rewriting $\left(  e_i\frac{p}{m_i} - e_j\frac{p'}{m_j} \right)$ as 
$$
\left(  e_i\frac{p}{m_i} - e_j\frac{p'}{m_j} \right) =\left(\frac{p}{m_i} +\frac{p'}{m_j}\right)\frac{(e_i-e_j)}{2} + \left(\frac{p}{m_i} -\frac{p'}{m_j}\right)\frac{(e_i+e_j)}{2}.
$$
\\

It is easy to see that, for $i,j=1,\ldots,N$, the matrix
$$ \mathscr{A}^{(ij)}\equiv \iint M_i M_j' A^{(ij)}\left[\frac{p}{m_i}-\frac{p'}{m_j}\right]dp dp' $$
is positive definite, while 
$$ \mathscr{B}^{(ij)}\equiv \frac{1}{4}\iint M_i M_j' \left(\frac{p}{m_i}+\frac{p'}{m_j}\right)\cdot
 A^{(ij)}\left[\frac{p}{m_i}-\frac{p'}{m_j}\right]\left(\frac{p}{m_i}+\frac{p'}{m_j}\right) dp dp' > 0. $$
We conclude:
\begin{lemma}\label{lem.fLbf.2}
There exists an explicitly computable constant $C_2>0$ such that:
\begin{align*}
 -(f^\parallel,L^b f^\parallel)_{L^2_p} &\geq C_2 \sum_{i,j=1}^N\left( |u_i-u_j|^2 + (e_i-e_j)^2 \right),
 \qquad f\in D(L), 
\end{align*}
where the $p-$independent quantities $u_i$, $e_i$ are related to $f$ through \eqref{f.par}.
\end{lemma}

The last step in the proof of the spectral gap for $L$ is the result shown in the next lemma. 
\begin{lemma}\label{lem.ue}
 There exists an explicitly computable constant $C_3>0$ such that
 \begin{align}
  \sum_{i,j=1}^N\left( |u_i-u_j|^2 + (e_i-e_j)^2 \right)\geq C_3\left( 
  \|f - \Pi^L f\|_{\cH}^2 - 2\|f^\perp\|_{\cH}^2 \right),\qquad f\in D(L),\label{est.ue}
 \end{align}
 where the $p-$independent quantities $u_i$, $e_i$ are related to $f$ through \eqref{f.par}.
\end{lemma}

For the proof of Lemma \ref{lem.ue} we refer directly to the one of Lemma 15 in \cite{DauJueMouZam}. In such lemma the authors prove (\ref{est.ue}) for $f$ solution to a multi-species linearized Boltzmann operator. The proof only relies on the structure of $N(L)$ and $N(L^m)$, which is the same in both multi-species Boltzmann system studied in \cite{DauJueMouZam} and the Landau systems considered in this manuscript.\\

Summarizing, Lemmas \ref{lem.fLf.1}, \ref{lem.fLbf.2} and \ref{lem.ue} imply that for every $f\in D(L)$ 
\begin{align*}
 -(f,Lf)_{L^2_p}\geq \frac{\eta}{2}C_2 C_3\|f - \Pi^L f\|_{\cH}^2 + (\lambda_m - \eta(C_1 + C_2 C_3))\|f^\perp\|_{\cH}^2.
\end{align*}
Choosing $\eta = \min\{1, \lambda_m/(C_1 + C_2 C_3)\}$ we obtain the desired spectral gap with
\begin{align*}
 \lambda = \frac{C_2 C_3}{2}\min\left\{1 , \frac{\lambda_m}{C_1 + C_2 C_3}\right\}.
\end{align*}
This finishes the proof of Theorem ~\ref{thr.spec.gap}.

\section{Exponential decay to global equilibrium}\label{Exponential decay to global equilibrium}
This section is devoted to Theorem ~\ref{thr.conv}. 
The proof relies on the spectral
gap of Theorem ~\ref{thr.spec.gap} and on the hypocoercivity method by Mouhot and Neumann \cite{MouNeu}. 
We have to show that there exists a suitable decomposition of $L$ for which conditions $(i)-(iv)$ in Theorem \ref{MN_Theorem} hold. 
 
We preliminarily observe that $L$ is bounded w.r.t. the $\cH$ norm, that is:
\begin{align}
 |(f,Lg)_{L^2_p}|\leq C\|f\|_\cH\|g\|_\cH , \qquad f,g\in D(L).\label{MN.1}
\end{align}
Relation \eqref{MN.1} can be showed by arguing as in the proof of \eqref{intwo}.

Using formulation \eqref{Lij.b}, the operator $L$ can be rewritten as $L=K-\Lambda$ with:
\begin{align}\label{K.1}
 K_{i}(f) &:= 
 -\frac{1}{\sqrt{M_i}}\sum_{j=1}^N\Div_p\int M_i M_j' A^{(ij)}\left[\frac{p}{m_i}-\frac{p'}{m_j}\right]\na\left( \frac{f_j'}{\sqrt{M_j'}}\right)dp',\\
 \nonumber
 \Lambda_i(f) &:=
 -\frac{1}{\sqrt{M_i}}\sum_{j=1}^N\Div_p\int M_i M_j' A^{(ij)}\left[\frac{p}{m_i}-\frac{p'}{m_j}\right]\na\left( \frac{f_i}{\sqrt{M_i}}\right)dp' .
\end{align}
For the operator $\Lambda$ we will use the following estimates proven by Guo in \cite[{Lemma 3, Lemma 6}]{Guo}: for each $ f\in D(L)$ we have
\begin{align}
  c_1\|f\|_\cH^2\leq (f,\Lambda f)_{L^2_p}\leq & c_2\|f\|_\cH^2, \label{MN.2}\\
  (D_{x}^{\alpha}D_{p}^{\beta} f, D_{x}^{\alpha}D_{p}^{\beta}\Lambda f)_{L^2(\mathbb{T}^3\times\R^3)}\geq &c_3\|D_{x}^{\alpha}D_{p}^{\beta} f\|_{L^2(\mathbb{T}^3, \cH)}^2  - \eta\sum_{\substack{\bar{\beta} \le \beta\\ \bar\beta\neq\beta}} \|D_{x}^{\alpha}D_{p}^{\bar{\beta}} f\|_{L^2(\mathbb{T}^3, \cH)}^2\nonumber \\
& -c(\eta)\|M D_x^\alpha f\|_{L^2(\mathbb{T}^3\times\R^3)}^2, \label{MN.3}
\end{align}
with $\eta>0$ arbitrary and $c(\eta)$ a positive constant depending only on $\eta$. 
We point out that assumption {\em(ii')} in Theorem \ref{MN_Theorem}
follows straightforwardly from \eqref{MN.3} by choosing $\eta<c_3$ and noticing that all the terms containing derivatives of $f$ of order strictly lower than $k$ can be trivially controlled
by the $H^{k-1}$ norm of $f$. In particular $\|M D_x^\alpha f\|_{L^2(\R^3)}^2\leq C \|D_x^\alpha f\|_{L^2(\R^3)}^2\leq C\|f\|_{H^{k-1}(\mathbb{T}^3\times\R^3)}^2$
since $|\beta|>0$ and $|\alpha|+|\beta|\leq k$.

Concerning $K$, we need the following lemma which proves at the same time $(iii)$ and  $(iii')$ of Theorem \ref{MN_Theorem}: 

\begin{lemma}
For every $\delta>0$ there exists a constant $C(\delta)>0$ such that for $|\alpha|+|\beta|\le k$ with $k\ge 4$ and $\beta\ge 1$ :
\begin{align}
 (D_{x}^{\alpha}D_{p}^{\beta} f, D_{x}^{\alpha}D_{p}^{\beta} K f)_{L^2(\T\times\R^{3})}\leq 
 \delta\| D_{x}^{\alpha}D_{p}^{\beta} f\|_{L^2(\T\times\R^{3})}^2 + C(\delta)\|f\|_{H^{k-1}(\T\times\R^{3})}^2.\label{MN.4}
\end{align}
\end{lemma}
\begin{proof}
We first observe that 
\begin{align*}
 \na M_i = -\frac{p}{m_i k_{B}T}M_i.
\end{align*}
Then $K$ can be rewritten as:
\begin{align}\label{K.3}
 K_{i,j}(f) =& 
 -\frac{1}{\sqrt{M_i}}\int \Div_p \left(M_i  A^{(ij)}\left[\frac{p}{m_i}-\frac{p'}{m_j}\right]\right)M_j'\na\left( \frac{f_j'}{\sqrt{M_j'}}\right)dp'  \nonumber\\
  \nonumber
 =& \int \omega^{(ij)}\cdot \sqrt{M_j'}\na\left( \frac{f_j'}{\sqrt{M_j'}}\right)dp' \\
 =&\int \omega^{(ij)}\cdot\left(\na f_j' + f_j'\frac{p'}{2m_j k_{B}T}\right)dp',
\end{align}

with the kernel $\omega^{(ij)}$ defined as:
\begin{align*}
 \omega^{(ij)} := \sqrt{M_i M_j'}\left(
 A^{(ij)}\left[\frac{p}{m_i}-\frac{p'}{m_j}\right]\frac{p}{m_i k_{B}T}
 + \frac{2C^{(ij)}}{m_i}\left|\frac{p}{m_i}-\frac{p'}{m_j}\right|^\gamma \left(\frac{p}{m_i}-\frac{p'}{m_j}\right)
 \right).
\end{align*}
It is useful to estimate $\omega^{(ij)}$ and its Jacobian. 
Since 
$$ 
| A^{(ij)}\left[z\right] v | \le C^{(ij)} |z|^{\gamma+2} {|v|},
$$
 we have 
\begin{align}
 \label{est.om}
 |\omega^{(ij)}| &\leq \sqrt{M_i M_j'} \left( \frac{|p|}{m_i k_{B}T} + \frac{2C^{(ij)}}{m_i}\right)
 \left( \left|\frac{p}{m_i}-\frac{p'}{m_j}\right|^{\gamma+2} + \left|\frac{p}{m_i}-\frac{p'}{m_j}\right|^{\gamma+1} \right).
 \end{align}
Taking into account that the magnitude of the derivative of every element of $A^{(ij)}\left[\frac{p}{m_i}-\frac{p'}{m_j}\right]$ w.r.t. each component of $p$ is bounded by $  C\left|\frac{p}{m_i}-\frac{p'}{m_j}\right|^{\gamma+1}$, and 
\begin{align*}
 \na\sqrt{M_i}= -\sqrt{M_i}\frac{p}{2 m_i k_{B}T}, 
\end{align*}
for some suitable polynomial $q(|p|)$ we have that the Jacobian of $\omega^{(ij)}$ with respect to $p$ can be estimated as 
\begin{align}
 \label{est.naom}
 |\na_p \otimes\omega^{(ij)}| &\leq \sqrt{M_i M_j'}q(|p|)
 \left( \left|\frac{p}{m_i}-\frac{p'}{m_j}\right|^{\gamma} + \left|\frac{p}{m_i}-\frac{p'}{m_j}\right|^{\gamma+2} \right).
\end{align}

Let us now introduce an arbitrary parameter $\eps>0$ and a cutoff function $\psi_\eps : [0,\infty)\to [0,1]$ such that
$\psi_\eps\in C^1([0,\infty))$, $\psi_\eps(x) = 1$ for $0\leq x\leq \eps$, $\psi_\eps(x)=0$ for $x\geq 2\eps$, 
$|\psi'_\eps|\leq C\eps^{-1}\chi_{(0,2\eps)}$. 
Moreover let us define $\Psi_\eps^{(ij)}(p,p') = \psi_\eps\left( \left|\frac{p}{m_i}-\frac{p'}{m_j} \right|\right)$.

We write $K = K^{(I)} + K^{(II)}$, where:
\begin{align*}
 K^{(I)}_i(f) &= \sum_{j=1}^N\int \left(1-\Psi_\eps^{(ij)} \right)\omega^{(ij)}\cdot\left(\na f_j' + f_j'\frac{p'}{2m_j k_{B}T}\right)dp',\\
 K^{(II)}_i(f) &= \sum_{j=1}^N\int \Psi_\eps^{(ij)}\omega^{(ij)}\cdot\left(\na f_j' + f_j'\frac{p'}{2m_j k_{B}T}\right)dp'.
\end{align*}
The function $\omega^{(ij)}$ is smooth in the region $\{|p/m_i-p'/m_j|>2\eps\}$, thus
$$
(1+|p'|)D_{p}^{2\beta}\left( \left(1-\Psi_\eps^{(ij)} \right)\omega^{(ij)} \right)\in L^\infty_{p,p'}.
$$ 
From Young's inequality and the fact that 
\begin{align*}
\|D_v^1 D_{x}^{\alpha}f\|^2_{L^2(\T\times\R^{3})} &\le C(\|D_{x}^{\alpha}D_{p}^{\beta} f\|_{L^2(\T\times\R^{3})}^2 + \|f\|_{L^2(\T\times\R^{3})}^2), \\
\| D_{x}^{\alpha}f\|^2_{L^2(\T\times\R^{3})} &\le  \|f\|_{H^{k-1}(\T\times\R^{3})}^2,
\end{align*}
we get
\begin{align}\label{fKf.I}
 (D_{x}^{\alpha}D_{p}^{\beta} f, &D_{x}^{\alpha}D_{p}^{\beta} K^{(I)} f)_{L^2(\T\times\R^{3})} \\
 \nonumber
 &= \sum_{i,j=1}^N\iiint D_{x}^{\alpha}D_{p}^{\beta} f_i\cdot
 \left(D_{p}^{\beta}\left( \left(1-\Psi_\eps^{(ij)} \right)\omega^{(ij)} \right)\right)
 \left(\na_{p'} D_{x}^{\alpha}f_j' + D_{x}^{\alpha}f_j'\frac{p'}{2m_j k_{B}T}\right)dp dp' dx\\
 \nonumber
 &= (-1)^{|\beta|}\sum_{i,j=1}^N\iiint D_{x}^{\alpha}f_i
 \left(D_{p}^{2\beta}\left( \left(1-\Psi_\eps^{(ij)} \right)\omega^{(ij)} \right)\right)
 \left(\na_{p'} D_{x}^{\alpha} f_j' + D_{x}^{\alpha} f_j'\frac{p'}{2m_j k_{B}T}\right)dp dp' dx\\
 \nonumber
 &\leq C(\eps)\|D_{x}^{\alpha}f\|_{L^2(\T\times\R^{3})}\left( \|\na_{p} D_{x}^{\alpha}f\|_{L^2(\T\times\R^{3})} + \|D_{x}^{\alpha}f\|_{L^2(\T\times\R^{3})} \right)\\
 &\leq \delta\|D_{x}^{\alpha}D_{p}^{\beta} f\|_{L^2(\T\times\R^{3})}^2 + \delta^{-1}C(\eps)\|f\|_{H^{k-1}(\T\times\R^{3})}^2 .\nonumber
\end{align}
We write $\beta = \hat\beta + \xi$ with $|\hat\beta|=1$, $|\xi|=k-1$, so that $D_{p}^{\beta}=D_{p}^{\xi}D_{p}^{\hat\beta}$.
Let us compute the term
\begin{align}\label{Dkappa}
D_{p}^{\beta}K^{(II)}(f) = D_{p}^{\xi}\sum_{j=1}^N\int \Theta^{ij}_{\eps,\hat\beta}[p,p']\cdot\left(\na f_j' + f_j'\frac{p'}{2m_j k_{B}T}\right)dp',
\end{align}
with 
$$ \Theta^{ij}_{\eps,\hat\beta}[p,p'] := D_{p}^{\hat\beta}(\Psi_\eps^{(ij)}\omega^{(ij)})[p,p']. $$
By making the transformation $p'/m_{j}\mapsto p/m_{i} - p'/m_{j}$ inside the integral in \eqref{Dkappa} we obtain
\begin{align}\label{Dkappa.2}
D_{p}^{\beta}K^{(II)}(f) &= D_{p}^{\xi}\sum_{j=1}^N\int \Theta^{ij}_{\eps,\hat\beta}[p,(m_{j}/m_{i})p - p']\cdot
\left(\na f_j\left[\frac{p}{m_{i}} - \frac{p'}{m_{j}}\right]\right.\\ 
&\left. + f_j\left[\frac{p}{m_{i}} - \frac{p'}{m_{j}}\right]\frac{1}{2 k_{B}T}
\left( \frac{p}{m_{i}} - \frac{p'}{m_{j}} \right)
\right)dp'.\nonumber
\end{align}
Let us estimate first the expression
\begin{align*}
&(|p|+|p'|)|\Theta^{ij}_{\eps,\hat\beta}[p,(m_{j}/m_{i})p - p']| = (|p|+|p'|)|D_{p}^{\hat\beta}(\Psi_\eps^{(ij)}\omega^{(ij)})[p,(m_{j}/m_{i})p - p']|\\
&\leq (|p|+|p'|)|D_{p}^{\hat\beta}(\Psi_{\eps}^{(ij)})[p,(m_{j}/m_{i})p - p']| |\omega^{(ij)}[p,(m_{j}/m_{i})p - p']| \\
&\qquad + (|p|+|p'|)|\Psi_{\eps}^{(ij)}[p,(m_{j}/m_{i})p - p']| |D_{p}^{\hat\beta}(\omega^{(ij)})[p,(m_{j}/m_{i})p - p']| .
\end{align*}
By using \eqref{est.om}, \eqref{est.naom} and the properties of the cutoff $\Psi_{\eps}^{(ij)}$ we deduce
\begin{align}\label{est_fin}
& (|p|+|p'|)|\Theta^{ij}_{\eps,\hat\beta}[p,(m_{j}/m_{i})p - p']|\leq C \left( |p'|^{\gamma} + |p'|^{\gamma+1} + |p'|^{\gamma+2} \right)
\chi_{\{|p'|\le 2\eps m_{j}\}}
\end{align}
for some constant $C>0$. Since the local singularities of $\Theta^{ij}_{\eps,\hat\beta}[p,(m_{j}/m_{i})p - p']$ only depend on $p'$ (after the change of variable $p'/m_{j}\mapsto p/m_{i} - p'/m_{j}$),  the estimate in (\ref{est_fin}) 
holds also for the derivatives of $\Theta^{ij}_{\eps,\hat\beta}[p,(m_{j}/m_{i})p - p']$ with respect to $p$, i.e.
\begin{align}\label{awful}
& (|p|+|p'|)|D_{p}^{\xi_{0}}\Theta^{ij}_{\eps,\hat\beta}[p,(m_{j}/m_{i})p - p']|\leq C\phi_{j,\eps}(p')\quad 0\leq\xi_{0}\leq \xi , \\
& \phi_{j,\eps}(p')\equiv \left( |p'|^{\gamma} + |p'|^{\gamma+1} + |p'|^{\gamma+2} \right)
\chi_{\{|p'|\le 2\eps m_{j}\}}.\label{awful.2}
\end{align}
Furthermore, assumption $\gamma\geq -2$ implies
\begin{align}
\|\phi_{j,\eps}\|_{L^{1}(\R^{3})} \leq C(\eps^{\gamma+3}+\eps^{\gamma+4}+\eps^{\gamma+5})\leq C\eps.\label{awful.3}
\end{align}
From \eqref{awful}, \eqref{awful.2} it follows (recall that $K^{(II)}$ does not depend on $x$)
$$ |D^{\alpha}_{x}D^{\beta}_{p}K^{(II)}(f)|\leq C \sum_{0\leq\beta'\leq\beta} \phi_{j,\eps}\ast |D_{x}^{\alpha}D_{p}^{\beta'}f| . $$
As a consequence, thanks to \eqref{awful.3},
\begin{align*}
\| D^{\alpha}_{x}D^{\beta}_{p}K^{(II)}(f)\|_{L^{2}(\T\times\R^{3})} &\leq C\left\| \phi_{j,\eps}\right\|_{L^{1}(\R^{3})} 
\sum_{0\leq\beta'\leq\beta}\left\| D_{x}^{\alpha}D_{p}^{\beta'}f \right\|_{L^{2}(\T\times\R^{3})}\\
&\leq C \eps \sum_{0\leq\beta'\leq\beta}\left\| D_{x}^{\alpha}D_{p}^{\beta'}f \right\|_{L^{2}(\T\times\R^{3})},
\end{align*}
from which it follows
\begin{align}\label{fKf.II}
 (D_{x}^{\alpha}D_{p}^{\beta} f, D_{x}^{\alpha}D_{p}^{\beta} K^{(II)} f)_{L^2(\T\times\R^{3})}\leq 
 C\eps\|f\|_{H^{k}(\T\times\R^{3})}^2 . 
\end{align}
Since $\eps>0$ is arbitrary, from \eqref{fKf.I}, \eqref{fKf.II} the statement \eqref{MN.4} follows. This finishes the proof.

\end{proof}
Relations \eqref{MN.1}--\eqref{MN.4} and the spectral gap  allow us to apply Theorem \ref{MN_Theorem}, which yields (\ref{conv}).\\

We now show the second part of Theorem \ref{thr.conv}. The non-linear terms $\Gamma_i(f,f)$, defined as 
$$ \Gamma_i(f,f) = \frac{1}{\sqrt{M_i}}\sum_{j=1}^{N}Q_{ij}\left( \sqrt{M_i}f_i , \sqrt{M_j}f_j \right): =\sum_{j=1}^{N} {\Theta}_i(f_i,f_j),$$
with 
 \begin{align*}
{\Theta}_i(f_i,f_j)
=&     \Div_p \left( \int A^{(ij)}\left[\frac{p}{m_i}-\frac{p'}{m_j}\right]\sqrt{M_j'}f_j'dp'  \cdot \nabla f_i \right)  \\ 
&-   \Div_p \left( f_i \int A^{(ij)}\left[\frac{p}{m_i}-\frac{p'}{m_j}\right]\sqrt{M_j'}\nabla f_j'dp'   \right) \\
& -  \int A^{(ij)}\left[\frac{p}{m_i}-\frac{p'}{m_j}\right] \frac{p'}{m_j}\sqrt{M_j'} f_j'dp' \cdot \nabla f_i \\
&  +  f_i  \int A^{(ij)}\left[\frac{p}{m_i}-\frac{p'}{m_j}\right] \frac{p'}{m_j}\sqrt{M_j'} \cdot \nabla f_j'dp' .
  \end{align*} 
We now recall an estimate by Guo in \cite[Thr.~3]{Guo} which states that the inner product $ (\Theta_i(f_i,f_j),f_i)_{H^k_{x,p}}$ can be bounded by the $H^k_{x,p}$ and $H^k_{x}\mathcal{H}$ norms of $f_i$ and $f_j$; more precisely 
\begin{align*}
(\Theta_i(f_i,f_j),f_i)_{H^k_{x,p}} \le C \left( \| f_i\|_{H^k_{x,p}} \|f_j\|_{H^k_{x}\mathcal{H}} + \| f_j\|_{H^k_{x,p}} \|f_i\|_{H^k_{x}\mathcal{H}}\right)  \|f_i\|_{H^k_{x}\mathcal{H}}.
\end{align*}
Therefore 
$$(\Gamma_i(f,f),f_i)_{H^k_{x,p}}  \le C \| f_i\|_{H^k_{x,p}}  \|f_i\|_{H^k_{x}\mathcal{H}} \left( \sum_{i=1}^N \|f_j\|_{H^k_{x}\mathcal{H}}\right) +   \|f_i\|^2_{H^k_{x}\mathcal{H}}\left( \sum_{i=1}^N \| f_j\|_{H^k_{x,p}} \right) ,
$$
which implies
\begin{align}\label{mixed_term}
(\Gamma (f,f),f)_{H^k_{x,p}}  := \sum_{i=1}^N (\Gamma_i(f,f),f_i)_{H^k_{x,p}} \le C   \| f\|_{H^k_{x,p}} \|f\|^2_{H^k_{x}\mathcal{H}}.
\end{align}
Define now the function $f:= \frac{F-\mathcal{M}}{\sqrt{\mathcal{M}}}$ with $\mathcal{M}(p) $ and $F$  respectively the unique equilibrium state and the unique smooth solution to (\ref{eq}). The function $f = (f_1, f_2,...,f_N)$ solves
\begin{equation*}
                         \pa_t f_i + \frac{p}{m_i} \cdot \nabla_x f_i =   \sum_{j=1}^N L_{ij}(f_i,f_j) + \Gamma_i(f_i,f_j). %
\end{equation*}
Thanks to Theorem ~\ref{thr.spec.gap} and (\ref{mixed_term}) one can deduce
\begin{align*} 
\frac{1}{2}\pa_t \| f \|_{H^k_{x,p}}^2 \le -\lambda  \| f \|_{H^k_{x}\mathcal{H}}^2 + C  \| f\|_{H^k_{x,p}} \|f\|^2_{H^k_{x}\mathcal{H}}.
\end{align*}
The above differential inequality can be solved by simple iteration method: since $ \| \fin \|_{H^k_{x}\mathcal{H}} \le \varepsilon$, there exists a positive time $T_0$ such that $\| f\|_{H^k_{x,p}} \le 2\varepsilon$ for all $t\in [0,T_0]$. Hence any solution to 
\begin{align*}
\frac{1}{2}\pa_t \| h \|_{H^k_{x,p}}^2 = -\frac{\lambda}{2}  \| h \|_{H^k_{x}\mathcal{H}}^2, \quad \| h_{\textrm{in}} \|_{H^k_{x}\mathcal{H}} = \varepsilon,
\end{align*}
satisfies $  \| f \|_{H^k_{x,p}}^2 \le  \| h \|_{H^k_{x,p}}^2 \le \varepsilon e^{-\lambda/2 t}$ for $t\in [0,T_0]$, taking into account that the $H^k_{x}\mathcal{H}$-norm controls the ${H^k_{x,p}}$-norm. At time $T_0$ we can restart the same process since $ \| f(\cdot, T_0) \|_{H^k_{x}\mathcal{H}} \le \varepsilon$. This finishes the proof of Theorem ~\ref{thr.conv}.

\section{Appendix}


\begin{lemma}\label{lem.K}
 The operator $K : L^2_p\to L^2_p$ defined in \eqref{K.1} is compact.
\end{lemma}
\begin{proof}
We will show that $K$ is the limit, in the operator norm, of a sequence of Hilbert-Schmidt operators.
From \eqref{K.3} it follows:
\begin{align*}
 K_i(f) &= \sum_{j=1}^N\int k^{(ij)}(p,p')f_j(p')dp',\qquad
 k^{(ij)}(p,p') = \frac{p'}{m_j k_{B}T}\cdot\omega^{(ij)} - \Div_{p'}\omega^{(ij)} .
\end{align*}
The following estimate is a consequence of \eqref{est.om} and \eqref{est.naom}:
\begin{align*}
 |k^{(ij)}(p,p')|
 &\leq C \left(M_i(p)M_j(p')\right)^{1/4}
 \left( \left|\frac{p}{m_i}-\frac{p'}{m_j}\right|^{\gamma} + \left|\frac{p}{m_i}-\frac{p'}{m_j}\right|^{\gamma+2} \right)\\
 &\leq C W\left( \frac{p}{m_i}-\frac{p'}{m_j} \right),\nonumber\\
 W(z) &\equiv e^{-\delta |z|^2}\left( |z|^{\gamma} + |z|^{\gamma+2} \right),
\end{align*}
for some suitable constant $\delta>0$. 

Let $\xi_n$ be the characteristic function of the ball $B\left(0,\frac{1}{n}\right)$, and let us define the sequence of operators 
$K^{(n)} = (K^{(n)}_1,\ldots,K^{(n)}_N) : L^2_p\to L^2_p$,
\begin{align*}
 K^{(n)}_i(f) &= \sum_{j=1}^N\int k_n^{(ij)}(p,p')f_j(p')dp',\\
 \nonumber
 k_n^{(ij)}(p,p') &= k^{(ij)}(p,p')\left(1-\xi_n\left( \frac{p}{m_i}-\frac{p'}{m_j} \right)\right).
\end{align*}
It is clear that $k_n^{(ij)}\in L^2_{p,p'}$, so $K^{(n)}$ is a Hilbert-Schmidt operator for all $n\in\N$. In particular $K^{(n)}$ is compact.
Let us now estimate:
\begin{align*}
 \left|K_i(f) - K_i^{(n)}(f)\right| &\leq \int |k^{(ij)}(p,p')|\xi_n\left( \frac{p}{m_i}-\frac{p'}{m_j} \right) |f_j(p')|dp'\\
 &\leq\sum_{j=1}^N \left( W\xi_n \right)\ast f_j .
\end{align*}
It follows:
\begin{align*}
 \frac{\|K(f) - K^{(n)}(f)\|_{L^2}}{\|f\|_{L^2}}\leq C\| W\xi_n\|_{L^1}
 = C\int_{\{|z|<1/n\}}e^{-\delta |z|^2}\left( |z|^{\gamma} + |z|^{\gamma+2} \right)dz\leq \frac{C}{n},
\end{align*}
since $\gamma+2\geq 0$. This means that $K^{(n)}\to K$ strongly in $\mathscr{L}(L^2_p)$, which implies that $K$ is compact.
This finishes the proof.
\end{proof}

\end{document}